\documentclass{article}
\usepackage{amsmath,amsthm,amssymb}
\usepackage{algorithm}
\usepackage{algpseudocode}
\usepackage{array}      
\usepackage{booktabs}   
\usepackage{caption}
\usepackage{color}    
\usepackage{amsmath}
\usepackage{amssymb}
\usepackage{tabularx}
\usepackage{longtable}

\newtheorem{Theorem}{Theorem}[section]
\newtheorem{Lemma}[Theorem]{Lemma}
\newtheorem{Remark}[Theorem]{Remark}

\newtheorem{Corollary}[Theorem]{Corollary}

\newtheorem{Example}[Theorem]{Example}
\numberwithin{equation}{section}
\usepackage[latin1]{inputenc}

\begin{document}
\title{{\LARGE New Bounds for Linear Codes with Applications}}

\author{Liren Lin$^1$, Guanghui~Zhang$^2$, Bocong Chen$^3$ and Hongwei Liu$^4$\footnote{E-mail addresses:
{\it l\underline{~}r\underline{~}lin86@163.com (L. Lin); zghui@squ.edu.cn (G. Zhang);
bocongchen@foxmail.com (B. Chen); hwliu@ccnu.edu.cn (H. Liu).}}
}

\date{\small
$1.$  Hubei Key Laboratory of Applied Mathematics, School of Cyber Science and Technology, Hubei University, Wuhan, Hubei 430062, China\\
$2.$ School of Mathematics and Physics, Suqian University, Suqian, Jiangsu 223800, China\\
$3.$ School of Mathematics, South China University of Technology, Guangzhou 510641, China\\
$4.$ School of Mathematics and Statistics,
Central China Normal University,
Wuhan,   430079, China
}


\maketitle

\begin{abstract}
Bounds on linear codes play a central role in coding theory,
as they capture the fundamental trade-off between error-correction capability (minimum distance)
and information rate (dimension relative to length).
Classical results characterize this trade-off solely in terms of the parameters $n$, $k$, $d$ and $q$.
In this work we derive new bounds under the additional assumption that
the code contains a nonzero codeword of weight $w$.
By combining residual-code techniques with classical results such
as the Singleton and Griesmer bounds,
we obtain explicit inequalities linking $n$, $k$, $d$, $q$ and $w$.
These bounds impose sharper restrictions on admissible codeword weights,
particularly those close to the minimum distance or to the code length.
Applications include refined constraints on the weights of MDS codes,
numerical  restrictions on general linear codes, and excluded weight ranges in the weight distribution.
Numerical comparisons  across standard parameter sets demonstrate that these
$w$-aware bounds strictly enlarge known excluded weight ranges and sharpen structural limitations on linear codes.

\medskip
	
	
\end{abstract}
	
\section{Introduction}
Let $\mathbb{F}_q$ denote the finite field of order $q$, where $q$ is a prime power.
For a positive integer $n$, we write $\mathbb{F}_q^n$ for the $n$-dimensional  coordinate space  over $\mathbb{F}_q$.
A linear subspace $\mathcal{C}\subseteq \mathbb{F}_q^n$ is called a
 \emph{linear code} over $\mathbb{F}_q$, and its elements are called \emph{codewords}.
The parameter $n$ is referred to as the \emph{length} of $\mathcal{C}$.
As a vector subspace of $\mathbb{F}_q^n$, the code $\mathcal{C}$ has
dimension $k$, called the \emph{dimension} of the code.
The \emph{Hamming distance} between two vectors $\mathbf{x},\mathbf{y}\in\mathbb{F}_q^n$ is the number of coordinate positions on which they differ, and is denoted by $d(\mathbf{x},\mathbf{y})$.
For a vector $\mathbf{x}\in\mathbb{F}_q^n$, the \emph{Hamming weight}, denoted by $\mathrm{wt}(\mathbf{x})$, is the number of nonzero coordinates of $\mathbf{x}$.
Clearly, for any $\mathbf{x},\mathbf{y}\in\mathbb{F}_q^n$ we have
\[
d(\mathbf{x},\mathbf{y}) = \mathrm{wt}(\mathbf{x}-\mathbf{y}).
\]
The \emph{minimum distance} of a code $\mathcal{C}$ is the minimum Hamming distance between any two distinct codewords of $\mathcal{C}$:
\[
d = \min_{\substack{\mathbf{c}_1,\mathbf{c}_2\in \mathcal{C}\\ \mathbf{c}_1\ne \mathbf{c}_2}} d(\mathbf{c}_1,\mathbf{c}_2).
\]
A linear code $\mathcal{C}$ over $\mathbb{F}_q$ with length $n$, dimension $k$, and minimum distance $d$ is called an $[n,k,d]_q$ code.
When the alphabet size $q$ is clear from context, we simply write $[n,k,d]$.

Bounds for linear codes are mathematical relationships that describe the fundamental limits among
 the code length $n$, dimension $k$,   minimum distance $d$ and code alphabet size $q$.
The study of such bounds constitutes a central topic in coding theory,
lying at the intersection of mathematics, computer science, and electrical engineering.
This is not merely a theoretical exercise, but a crucial endeavor that guides the design, analysis, and implementation of efficient and reliable digital communication systems.

One of the most well-known results is the {\em Singleton bound} \cite{Singleton}, which states that for an
  $[n,k,d]$ linear code $\mathcal{C}$ over $\mathbb{F}_q$,
  $$d \leq n - k + 1.$$
A linear code achieving equality in the Singleton bound is called a {\em maximum distance separable (MDS) code}.
A stronger inequality, which takes into account the field size
$q$, is the {\em Griesmer bound} \cite{Griesmer}.
For an $[n,k,d]$ linear code over $\mathbb{F}_q$ with $k\geq 1$, it
asserts that
$$n \geq \sum_{i=0}^{k-1} \left\lceil \frac{d}{q^i} \right\rceil,$$
where $\lceil r\rceil$ denotes the ceiling of a real number  $r$,
i.e.,   the smallest integer greater than or equal to $r$.
A linear code meeting this bound with equality is called a {\em Griesmer code}.
In fact, the Griesmer bound can be regarded as a refinement of the Singleton bound in the linear case.
These inequalities guide code design, benchmark best-known parameters,
and constrain what weight spectra are even feasible (see, e.g., see \cite{huffman2003}).

As coding theory advances, a key direction for strengthening classical bounds involves introducing an additional structural parameter $s$ to the fundamental parameters $n,k,d,q$, thereby deriving refined inequalities that relate all five: $n,k,d,q,s$.
Typical choices of $s$ include:
\begin{itemize}
  \item Covering radius $R$.  Sphere-covering and external-distance arguments give
        $$q^{\,n-k}\le \sum_{i=0}^{R}\binom{n}{i}(q-1)^i$$ and $\lfloor(d-1)/2\rfloor\le R\le s^\ast$,
        relating $R$ to the number $s^\ast$ of nonzero dual weights (see Delsarte \cite{Delsarte}).

  \item Number of nonzero weights. Delsarte's theory and later refinements bound how many distinct nonzero weights a linear (or cyclic/quasi-cyclic) code can have, constraining admissible spectra
        (see \cite{CZ,LCEL,Shi}).

  \item Locality parameters for LRCs. For all-symbol locality $r$,
        $d\le n-k-\lceil k/r\rceil+2$ (see \cite{Go}); for $(r,\delta)$-locality,
        $d\le n-k+1-(\lceil k/r\rceil-1)(\delta-1)$ (see \cite{Pra});
        for small alphabets, alphabet-dependent refinements further tighten the trade-off (see \cite{Ca});
        availability parameters provide additional constraints \cite{Luo25,Ra}.

  \item $w$-weight codewords. Assume that an $[n,k,d]$ linear code
  over $\mathbb{F}_q$ contains a nonzero codeword of weight $w$.
  Chen and Xie \cite{CX} derived a new bound (hereafter referred to as the Chen-Xie bound) on the code parameters involving $n,k,d,q$ and $w$.
 More significantly, they established excluded partial  weight distributions solely from the parameters $n,k,d$
 and $q$ which constitutes the first result of this kind.
\end{itemize}

In this paper we continue the line of investigation initiated in \cite{CX}.
Unlike the approach of \cite{CX}, which relies heavily on the theory of minimal vectors in linear codes \cite{AB},
 our method combines residual-code techniques with classical results such as the Singleton and Griesmer bounds.
 This yields explicit inequalities linking the parameters $n,k,d,q$ with a prescribed weight $w$.
 More explicitly, we have the following results.
\begin{enumerate}
  \item \textbf{Residual-Singleton bound (Theorem~\ref{firstbound}).}
        If an $[n,k,d]$ linear code over $\mathbb{F}_q$ contains a nonzero codeword of weight $w<\tfrac{qd}{q-1}$, then
        \[
          d \;\le\; n-k-\Big\lceil\tfrac{w}{q}\Big\rceil+2.
        \]
        This transforms the knowledge of one codeword weight into a direct tightening of
        the distance bound. The inequality is tight for several standard families illustrated in the paper.

  \item \textbf{Residual-Griesmer bound (Theorem~\ref{secondbound}).}
         If an $[n,k,d]$ linear code over $\mathbb{F}_q$ contains a nonzero codeword of weight $w<\tfrac{qd}{q-1}$ and $k\ge 2$, then
        \[
          n \;\ge\; d+\Big\lceil\tfrac{w}{q}\Big\rceil
               +\sum_{i=1}^{k-2}\Big\lceil \tfrac{\,d-w+\lceil \frac{w}{q}\rceil\,}{q^i}\Big\rceil,
        \]
        thereby quantifying the additional length forced by the presence of a codeword of weight $w$.
        Equality holds for several Griesmer codes given in our examples.

  \item \textbf{Global weight inequality (Theorem~\ref{2n-d}).}
        Every nonzero codeword satisfies $w\le q(n-d)$, which immediately implies $d\le \tfrac{q}{q+1}n$ by taking $w=d$.

  \end{enumerate}

  These results lead to refined constraints on the weights of MDS codes,
  new numerical conditions for general linear codes, and explicit vanishing ranges
  in the weight distribution. We focus especially on excluded weights,
  continuing the perspective of \cite{CX}.
  As direct consequences of Theorems~\ref{firstbound} and \ref{secondbound}, we obtain the following:

        \begin{itemize}
          \item \emph{Singleton Exclusion Criterion} (Theorem~\ref{firstvanishing1}): for $d<w<\frac{qd}{q-1}$, all weights with $w>q(n-k-d+2)$ vanish, yielding a long consecutive interval of forbidden weights.
          \item \emph{Griesmer Exclusion Criterion} (Theorem~\ref{firstvanishing3}): for $d<w<\frac{qd}{q-1}$, all weights with
          \[ n < d + \left\lceil \frac{w}{q} \right\rceil + \sum_{i=1}^{k-2}
\left\lceil \frac{d - w + \left\lceil \frac{w}{q} \right\rceil}{q^i} \right\rceil \]
vanish,
a stronger criterion that can exclude \emph{non-consecutive} weights.
Both theoretical analysis and numerical evidence (see Remark \ref{remark} and Tables 1-3)
demonstrate that our bounds improve upon the results of the Chen-Xie bound  \cite{CX} (see
Lemma \ref{knownresult} below).
\end{itemize}
In particular, we explicitly construct a binary $[11,3,6]$ code
whose set of nonzero weights is exactly $\{6,8\}$.
From the parameters $n,k,d,q$,
the Chen-Xie bound (Lemma~\ref{knownresult} below)
 guarantees the absence of codewords of weights $\{10,11\}$.
Applying Singleton Exclusion Criterion further excludes $\{9,10,11\}$.
Finally, Griesmer Exclusion Criterion establishes that $\{7,9,10,11\}$ vanish,
which is optimal, since the remaining non-excluded weights $\{6,8\}$ exactly match the true spectrum of the code.

This paper is organized as follows.
Section 2 develops the residual-code framework and proves the
$w$-aware Singleton and Griesmer bounds together with the global inequality
$w\leq q(n-d)$. Section 3 applies these results to MDS codes and to excluded weight ranges
(both Singleton- and Griesmer-based), with detailed examples and comparisons.
Section 4 concludes with broader implications and future research directions.

\section{New Bounds for Linear Codes}
We begin with the classical notion of a residual code due to Helgert and Stinaff \cite{HS}.
Let $\mathcal{C}$ be an $[n,k,d]$ linear code over $\mathbb{F}_q$.
For a nonzero codeword $\mathbf{c}\in\mathcal{C}$ with Hamming weight
$w=\mathrm{wt}(\mathbf{c})$, let $\mathcal{I}$ denote the set of coordinates
on which $\mathbf{c}$ is nonzero.
The \emph{residual code} of $\mathcal{C}$ with respect to $\mathbf{c}$,
denoted $\mathrm{Res}(\mathcal{C},\mathbf{c})$, is obtained by puncturing $\mathcal{C}$
at the coordinates in $\mathcal{I}$, and thus has length $n-w$.
The following standard lemma provides a lower bound on the minimum distance of residual codes;
see \cite{HN} or \cite[Theorem~2.7.1]{huffman2003}.

\begin{Lemma}\label{thm:residual}
Let $\mathcal{C}$ be an $[n,k,d]$ linear code
over $\mathbb{F}_q$ and let $\mathbf{c}\in\mathcal{C}$ be a nonzero codeword of weight $w<\frac{qd}{q-1}$.
Then $\mathrm{Res}(\mathcal{C},\mathbf{c})$ is an $[n-w,k-1,d']$ code with
$$
d'\ \ge\ d-w+\Big\lceil \frac{w}{q}\Big\rceil.
$$
\end{Lemma}

In this section, we establish a set of new bounds for linear codes,
which serve as the foundation for deriving sharper constraints on the weights of
 MDS codes, additional numerical conditions for general linear codes,
 and explicit vanishing intervals in their weight spectra.
The key idea is to exploit the weight $w$ of a known nonzero codeword
and pass to the corresponding residual code, to which classical bounds can be applied.
Applying the Singleton bound to residual codes yields the following constraint.

\begin{Theorem}[Residual-Singleton bound]\label{firstbound}
Let $\mathcal{C}$ be an $[n,k,d]$ linear code over $\mathbb{F}_q$
and let $\mathbf{c}$ be a nonzero codeword of weight $w < \frac{qd}{q-1}$. We have
$$d\leq n - k -\left\lceil \frac{w}{q} \right\rceil+ 2. $$
\end{Theorem}
\begin{proof}
By Lemma~\ref{thm:residual}, $\mathrm{Res}(\mathcal{C},\mathbf{c})$ has parameters $[n-w,k-1,d']$ with
$d'\ge d-w+\lceil \frac{w}{q}\rceil$.
The Singleton bound for $\mathrm{Res}(\mathcal{C},\mathbf{c})$ gives
\[
(n-w)-(k-1)+1\ \ge\ d'\ge d-w+\Big\lceil \frac{w}{q}\Big\rceil,
\]
whence
\[
n-k+2\ \ge\ d+\Big\lceil\frac{w}{q}\Big\rceil,
\]
as claimed.
\end{proof}

We present several examples to illustrate  Theorem \ref{firstbound}.
\begin{Example}{\rm
	The bound in Theorem~\ref{firstbound} is tight and can be attained by many
  familiar linear codes; see \cite{Ding2018}.
\begin{enumerate}
\item \textbf{Ternary Hamming $[13,10,3]$} (Example~2.37 in \cite{Ding2018}).
This ternary Hamming code has a codeword of weight $4$.
With $q=3$, $n=13$, $k=10$, $d=3$, and   $w=4$, we have
\[
n-k-\Big\lceil \tfrac{w}{q}\Big\rceil+2=13-10-\Big\lceil\tfrac{4}{3}\Big\rceil+2=3=d.
\]

\item \textbf{A ternary $[27,8,14]$ code} (Example~6.50 in \cite{Ding2018}).
This ternary code has weight enumerator $1+810z^{14}+702z^{15}+
1404z^{17}+780z^{18}+2106z^{20}+702z^{21}+54z^{26}+2z^{27}$.
Taking $w=20$ gives
\[
n-k-\Big\lceil \tfrac{w}{q}\Big\rceil+2=27-8-\Big\lceil\tfrac{20}{3}\Big\rceil+2=14=d.
\]

\item \textbf{Binary cyclic $[15,10,4]$} (Example~7.14 in \cite{Ding2018}).
 This binary cyclic code has weight enumerator $1+105z^4+280z^6+435z^8+168z^{10}+35z^{12}$.
With $w=6$, we have
\[
n-k-\Big\lceil \tfrac{w}{q}\Big\rceil+2=15-10-\Big\lceil\tfrac{6}{2}\Big\rceil+2=4=d.
\]
\end{enumerate}
Thus Theorem~\ref{firstbound} can be attained by diverse linear codes.
}
\end{Example}

Using the Griesmer bound on residual codes yields a refined length constraint.
\begin{Theorem}[Residual-Griesmer bound]\label{secondbound}
Let $\mathcal{C}$ be an $[n,k,d] $ linear code over $\mathbb{F}_q$ with $k\geq 2$,
and let $\mathbf{c}$ be a nonzero codeword of weight $w < \frac{qd}{q-1}$. Then we have
		$$n \geq d + \left\lceil \frac{w}{q} \right\rceil + \sum_{i=1}^{k-2} \left\lceil \frac{d - w + \left\lceil \frac{w}{q} \right\rceil}{q^i} \right\rceil.$$
\end{Theorem}
\begin{proof}
	Using Lemma~\ref{thm:residual} again,
the residual code $\text{Res}(\mathcal{C}, \mathbf{c})$ is an $[n-w, k-1, d']$ code with
	\[ d' \geq d - w + \left\lceil \frac{w}{q} \right\rceil. \]
Applying the Griesmer bound to $\text{Res}(\mathcal{C}, \mathbf{c})$, we have
	\[
	n - w \geq \sum_{i=0}^{k-2} \left\lceil \frac{d'}{q^i} \right\rceil.
	\]
Substituting the lower bound on $d'$ and separating the $i=0$ term we obtain the following:
	\[
	n - w \geq \left\lceil d' \right\rceil + \sum_{i=1}^{k-2} \left\lceil \frac{d'}{q^i} \right\rceil \geq \left(d - w + \left\lceil \frac{w}{q} \right\rceil\right) + \sum_{i=1}^{k-2} \left\lceil \frac{d - w + \left\lceil \frac{w}{q} \right\rceil}{q^i} \right\rceil.
	\]
Rearranging terms gives the inequality:
	\[
	n \geq d + \left\lceil \frac{w}{q} \right\rceil + \sum_{i=1}^{k-2} \left\lceil \frac{d - w + \left\lceil \frac{w}{q} \right\rceil}{q^i} \right\rceil,
	\]
which completes the proof.
\end{proof}
We illustrate, using two classical binary linear codes, that the Griesmer-type bound in Theorem~\ref{secondbound} provides a strictly sharper constraint than the Singleton-type bound in Theorem~\ref{firstbound}.

\begin{Example}
	{\rm
	Let $\mathcal{C}$ be the binary linear code with parameters $[15,5,7]$ which is taken from \cite[TABLE II]{CX}.
Consider a minimum-weight codeword of weight $w=7$.
	
	\begin{itemize}
		\item \textbf{Residual-Singleton bound (Theorem \ref{firstbound}):}
		\[
		d \leq n - k - \left\lceil\frac{w}{q}\right\rceil + 2 = 15 - 5 - 4 + 2 = 8.
		\]
		This is \emph{satisfied} ($7 < 8$) but not tight.
		
		\item \textbf{Residual-Griesmer bound (Theorem \ref{secondbound}):}
		\[
		\begin{aligned}
			n &\geq d + \left\lceil\frac{w}{q}\right\rceil + \sum_{i=1}^{k-2}\left\lceil\frac{d - w + \lceil \frac{w}{q} \rceil}{q^{i}}\right\rceil \\
			&= 7 + 4 + \left(2 + 1 + 1\right) = 15.
		\end{aligned}
		\]
		The equality $15 = 15$ shows this bound is \emph{tight}.
	\end{itemize}
	Hence, the \textbf{Residual-Griesmer bound} provides a strictly sharper constraint here.
}
\end{Example}

\begin{Example}
	{\rm
	Let $\mathcal{C}$ be the binary linear code $[31,5,16]$ which is taken from \cite[TABLE II]{CX}, and consider $w=16$.
	
	\begin{itemize}
		\item \textbf{Residual-Singleton bound (Theorem \ref{firstbound}):}
		\[
		d \leq 31 - 5 - \left\lceil\frac{16}{2}\right\rceil + 2 = 31 - 5 - 8 + 2 = 20.
		\]
		Satisfied ($16 < 20$) but not tight.
		
		\item \textbf{Residual-Griesmer bound (Theorem \ref{secondbound}):}
		\[
		\begin{aligned}
			n &\geq 16 + 8 + \sum_{i=1}^{3}\left\lceil\frac{16 - 16 + 8}{2^{i}}\right\rceil \\
			&= 16 + 8 + \left(4 + 2 + 1\right) = 31.
		\end{aligned}
		\]
		The equality $31 = 31$ confirms the bound is \emph{tight}.
	\end{itemize}
	Again, the \textbf{Residual-Griesmer bound} yields a sharper constraint.
}
\end{Example}

The preceding two bounds apply under the hypothesis $w<\tfrac{qd}{q-1}$; in
many applications the relevant $w$ lies close to the minimum distance $d$.
We now consider a scenario where the bound becomes tighter when
$w$ is near the code length
 $n$.

\begin{Theorem}\label{2n-d}
Let $\mathcal{C}$ be an $[n,k,d]$ linear code over
$\mathbb{F}_q$ with $k>1$. If $\mathcal{C}$ contains
 a nonzero codeword of weight $w$,
then
$$
w \le q\,(n-d).
$$
\end{Theorem}

\begin{proof}
Let $\mathbf{c}\in\mathcal{C}$ be a nonzero codeword of weight $\mathrm{wt}(\mathbf{c})=w$.
Split the coordinate set $[n]=\{1,\dots,n\}$ into
$$
S=\mathrm{supp}(\mathbf{c})=\big\{i\;\big|\; c_i\ne 0\}\quad\text{and}\quad T=[n]\setminus S=\{i\;|\; c_i=0\big\},
$$
giving $|S|=w$ and $|T|=n-w$.
Because $k>1$, the $1$-dimensional
subspace $\langle \mathbf{c}\rangle=\{\alpha\mathbf{c}\,|\,\alpha\in\mathbb{F}_q\}$
is a proper subset of $\mathcal{C}$.
Hence we can choose $\mathbf{u}\in\mathcal{C}\setminus\langle\mathbf{c}\rangle$.
For each $\alpha\in\mathbb{F}_q$, define the codeword
\[
\mathbf{v}_\alpha = \mathbf{u} + \alpha \mathbf{c}\in\mathcal{C}.
\]
If any $\mathbf{v}_\alpha$ were the zero vector, then $\mathbf{u}=-\alpha\mathbf{c}\in\langle\mathbf{c}\rangle$, contradicting the choice of $\mathbf{u}$. Thus \emph{all} $\mathbf{v}_\alpha$ are nonzero, giving
\begin{equation}\label{eq:lowerd}
\mathrm{wt}(\mathbf{v}_\alpha)\ \ge\ d\qquad\text{for every }\alpha\in\mathbb{F}_q.
\end{equation}
Fix $i\in S$, so $c_i\ne 0$. Consider the map
$$
\mathbb{F}_q\longrightarrow\mathbb{F}_q,\qquad \alpha\mapsto (\mathbf{v}_\alpha)_i \;=\; u_i+\alpha c_i.
$$
Since $c_i\ne 0$, this map is a bijection.
Consequently, as $\alpha$ runs over $\mathbb{F}_q$, the value $(\mathbf{v}_\alpha)_i$
takes the value $0$ for  exactly one  $\alpha$, and is nonzero for the remaining $q-1$ values.
For each fixed $\alpha$, let
$$
N_S(\alpha) = \bigl|\{i\in S\,\big|\,\; (\mathbf{v}_\alpha)_i\ne 0\}\bigr|
  \quad\text{(the number of nonzeros of $\mathbf{v}_\alpha$ on $S$).}
$$
It follows that
$$
N_S(\alpha)=w-\Big|\Big\{i\in S\,\Big|\,u_i+\alpha c_i=0\Big\}\Big|.
$$
Summing over $\alpha\in \mathbb{F}_q$ gives
$$
  \sum_{\alpha\in\mathbb{F}_q} N_S(\alpha)=qw-
  \sum_{\alpha\in \mathbb{F}_q}\Big|\Big\{i\in S\,\Big|\,u_i+\alpha c_i=0\Big\}\Big|.
$$
For each fixed $i\in S$ there is exactly one $\alpha\in \mathbb{F}_q$ such that $u_i+\alpha c_i=0$,
and we have
$$
 \sum_{\alpha\in \mathbb{F}_q}\Big|\Big\{i\in S\,\Big|\,u_i+\alpha c_i=0\Big\}\Big|=|S|=w.
$$
Therefore,
$$
  \sum_{\alpha\in\mathbb{F}_q} N_S(\alpha) \;=\; (q-1)\,|S| \;=\; (q-1)w.
$$
Hence the  average of $N_S(\alpha)$ over $\alpha$ equals $\frac{q-1}{q}w$,
so there exists $\alpha_0\in\mathbb{F}_q$ with
\begin{equation}\label{eq:aupper}
N_S(\alpha_0)\ \le\ \frac{q-1}{q}\,w.
\end{equation}
If $i\in T$ then $c_i=0$, so $(\mathbf{v}_\alpha)_i=u_i$ is independent of $\alpha$.
Let
$$
N_T= \bigl|\{i\in T\,\big|\,\; u_i\ne 0\}\bigr| \;=\; \bigl|\{i\in T\,\big|\,\; (\mathbf{v}_\alpha)_i\ne 0\}\bigr|
$$
for any  $\alpha$.
Trivially, $N_T\le |T|=n-w$.
For each $\alpha$ we have a disjoint decomposition
$$
\mathrm{wt}(\mathbf{v}_\alpha)\;=\;N_S(\alpha)\;+\; N_T.
$$
Apply this at $\alpha_0$ from Equation \eqref{eq:aupper} and use Equation \eqref{eq:lowerd}:
$$
d \ \le\ \mathrm{wt}(\mathbf{v}_{\alpha_0}) \ =\ N_S(\alpha_0) + N_T
\ \le\ \frac{q-1}{q}\,w \;+\; N_T
\ \le\ \frac{q-1}{q}\,w \;+\; (n-w)
\;=\; n-\frac{w}{q}.
$$
Rearranging gives \(w\le q(n-d)\), as claimed.
\end{proof}

\begin{Example}{\rm
Consider the first order binary Reed-Muller code ${\rm RM}(1,4)$,
which has parameters $[16,5,8]$. It is also known that ${\rm RM}(1,4)$ contains a codeword with weight $16$.
By setting $w=16$,  we observe that Theorem \ref{2n-d}  is satisfied with equality.
}
\end{Example}
\section{Applications}
In this section, we apply the bounds developed in Section~2.
The discussion is organized into three parts:
(i) restrictions for MDS codes and general linear codes from Theorems   \ref{firstbound} and \ref{2n-d},
(ii) vanishing weight ranges obtained via Theorem \ref{firstbound},
and (iii) strengthened vanishing weight ranges derived from Theorem \ref{secondbound}.
We substantiate the improvements through numerical comparisons,
showing that our results yield strictly sharper constraints on excluded weights than existing bounds.

\subsection{Applications to MDS codes and general linear codes}
As an immediate consequence of Theorem~\ref{firstbound}, we obtain the following restriction for MDS codes.
\begin{Corollary}\label{mds}
Let $\mathcal{C}$ be an $[n,k,d]$ MDS code over $\mathbb{F}_q$.
For any codeword of $\mathcal{C}$ with weight $w$ satisfying $d\leq w<\frac{qd}{q-1}$,
it must hold that $w\leq q.$
\end{Corollary}

Once the MDS conjecture is proven, Corollary \ref{mds} will follow trivially.
However, the conjecture remains unresolved. Nevertheless, the corollary may provide additional insights into MDS codes.

Taking $w=d$ in Theorem \ref{2n-d}, we have the next result.
\begin{Corollary}\label{cor:global-weight-ratio}
Let $\mathcal{C}$ be an $[n,k,d]$ linear code over
$\mathbb{F}_q$ with $k>1$. We have
$$
d\leq\frac{q}{q+1}n.
$$
\begin{proof}
Apply Theorem~\ref{2n-d} to a minimum-weight codeword ($w=d$) to obtain
$$
d \le \frac{q}{q+1}\,n.
$$
\end{proof}
\end{Corollary}

\begin{Example}{\rm
Let $\mathbb{F}_q=\{\alpha_1, \alpha_2, \cdots,\alpha_q\}$, and consider
the  $[q+1, 2, q]_q$ linear code over $\mathbb{F}_q$ with  generator matrix
$$
\begin{pmatrix}
1 & 1 & \cdots & 1 & 0\\
\alpha_1 & \alpha_2 &  \cdots & \alpha_q & 1
\end{pmatrix}.
$$	
Its parameters satisfy $d=\tfrac{q}{q+1}n$,
showing that equality in Corollary~\ref{cor:global-weight-ratio} can be attained.
}
\end{Example}

\subsection{Excluded weights via Theorem \ref{firstbound}}
The bounds established in Theorems~\ref{firstbound} and~\ref{secondbound} not only provide theoretical relationships among code parameters but also lead directly to powerful practical tools for determining which codeword weights cannot exist in a linear code with given parameters $n, k, d, q$. We formalize these applications into two explicit \emph{exclusion criteria}.

\begin{Theorem}[Singleton Exclusion Criterion]\label{firstvanishing1}
Let $\mathcal{C}$ be an $[n,k,d]$ linear code over $\mathbb{F}_q$ and
let $w$ be a positive integer satisfying $d \leq w < \frac{qd}{q-1}$.
If 	
$$
w > q(n - k - d + 2),
$$
then $\mathcal{C}$ has no codeword of weight $w$.
\end{Theorem}
\begin{proof}
Assume, for contradiction, that such a codeword exists.
By Theorem~\ref{firstbound},
$$d\leq n-k-\left\lceil \frac{w}{q} \right\rceil+2.$$
Therefore
$$\frac{w}{q}\leq \left\lceil \frac{w}{q} \right\rceil \leq n-k-d+2,$$
which implies
$$w\leq q(n-k-d+2),$$
a contradiction.
\end{proof}

To highlight the improvement, we compare our exclusion criteria with the Chen-Xie bound~\cite{CX}, which is state of the art in deducing vanishing intervals from $n,k,d$ and $q$.
For completeness, we restate their main result below.
\begin{Lemma}{\rm(Chen-Xie bound \cite{CX})}\label{knownresult}
	Let $\mathcal{C}$ be an $[n,k,d]$ linear code over $\mathbb{F}_q$.
If  an integer $v\geq0$ satisfies
	$$
	n-k+2<\frac{qd}{q-1}-v,
	$$
	then $\mathcal{C}$ has no codeword with weight $w$ in the interval  $$ \frac{qd}{q-1}-v-1\leq w \leq  \frac{qd}{q-1}-1,$$
	i.e. these weights are excluded.
\end{Lemma}

\begin{Remark}{\rm \label{remark}
		The Chen-Xie bound excludes the interval $[n-k+2, \left\lfloor\frac{qd}{q-1}\right\rfloor-1]$,
		whereas the Singleton Exclusion Criterion excludes $[q(n-k-d+2)+1, \left\lceil\frac{qd}{q-1}\right\rceil-1]$.
		Since it is obvious that the right endpoints satisfy that
$\left\lfloor\frac{qd}{q-1}\right\rfloor-1\leq \left\lceil\frac{qd}{q-1}\right\rceil-1$,
the comparison reduces to the left endpoints.
		The feasibility condition $(q-1)(n-k+2) < qd$ in the Chen-Xie bound implies $q(n-k-d+2) < n-k+2$. Therefore,
		 the Singleton Exclusion Criterion consistently yields strictly stronger exclusions for all codes satisfying this condition, as it excludes a larger interval of weights. This superiority is also clearly illustrated by the comparative data presented in Tables 1 and 2.
The data in Table 1 are taken from \cite[TABLE II]{CX}, whereas those in Table 2 are taken from \cite[TABLE III]{CX}.
For example, refer to the last row of Table 1, which presents a binary code with parameters $[93,5,48]_2$.
The data in the last row and second column showed that
$A_{93} = A_{92} = A_{91} = A_{90} = 0$, indicating that four weights vanish.
In contrast, applying our result from Theorem \ref{firstvanishing1}, we conclude that
$A_{93} = A_{92} =\cdots= A_{88} = A_{87} = 0$, meaning that seven weights vanish,
as shown in the last row and the final column of Table 1.
}
\end{Remark}

\begin{longtable}{|p{85.9pt}|p{100.4pt}|p{160.4pt}|}
	\caption{Comparison of excluded weight ranges for binary linear codes: Chen-Xie bound vs. Singleton Exclusion Criterion} \\
	\hline
	\textbf{Code parameters} & \textbf{Excluded weights} & \textbf{Excluded weights} \\
	$[n,k,d]_2$ & \textbf{(Chen-Xie bound)} & \textbf{(Singleton  Criterion)} \\
	\hline
	\endfirsthead
	
	\caption[]{Comparison of excluded weight ranges for binary linear codes: Chen-Xie bound vs. Singleton Exclusion Criterion (continued)} \\
	\hline
	\textbf{Code parameters} & \textbf{Excluded weights} & \textbf{Excluded weights} \\
	$[n,k,d]_2$ & \textbf{(Chen-Xie bound)} & \textbf{(Singleton  Criterion)} \\
	\hline
	\endhead
	
	\hline
	\endfoot
	
	\hline
	\endlastfoot
	
	$[15, 5, 7]_2$ & $13, 12$ & $13, 12, 11$ \\
	& (2 weights) & (3 weights) \\
	\hline
	$[21, 9, 8]_2$ & $15, 14$ & $15, 14, 13$ \\
	& (2 weights) & (3 weights) \\
	\hline
	$[31, 5, 16]_2$ & $31, 30, 29, 28$ & $31, 30, 29, 28, 27, 26, 25$ \\
	& (4 weights) & (7 weights) \\
	\hline
	$[32, 6, 16]_2$ & $31, 30, 29, 28$ & $31, 30, 29, 28, 27, 26, 25$ \\
	& (4 weights) & (7 weights) \\
	\hline
	$[47, 5, 24]_2$ & $47, 46, 45, 44$ & $47, 46, 45, 44, 43, 42, 41$ \\
	& (4 weights) & (7 weights) \\
	\hline
	$[48, 6, 24]_2$ & $47, 46, 45, 44$ & $47, 46, 45, 44, 43, 42, 41$ \\
	& (4 weights) & (7 weights) \\
	\hline
	$[55, 5, 28]_2$ & $55, 54, 53, 52$ & $55, 54, 53, 52, 51, 50, 49$ \\
	& (4 weights) & (7 weights) \\
	\hline
	$[56, 6, 28]_2$ & $55, 54, 53, 52$ & $55, 54, 53, 52, 51, 50, 49$ \\
	& (4 weights) & (7 weights) \\
	\hline
	$[59, 5, 30]_2$ & $59, 58, 57, 56$ & $59, 58, 57, 56, 55, 54, 53$ \\
	& (4 weights) & (7 weights) \\
	\hline
	$[60, 6, 30]_2$ & $59, 58, 57, 56$ & $59, 58, 57, 56, 55, 54, 53$ \\
	& (4 weights) & (7 weights) \\
	\hline
	$[61, 5, 31]_2$ & $61, 60, 59, 58$ & $61, 60, 59, 58, 57, 56, 55$ \\
	& (4 weights) & (7 weights) \\
	\hline
	$[62, 6, 31]_2$ & $61, 60, 59, 58$ & $61, 60, 59, 58, 57, 56, 55$ \\
	& (4 weights) & (7 weights) \\
	\hline
	$[63, 5, 32]_2$ & $63, 62, 61, 60$ & $63, 62, 61, 60, 59, 58, 57$ \\
	& (4 weights) & (7 weights) \\
	\hline
	$[63, 6, 32]_2$ & $63, 62, 61, 60, 59$ & $63, 62, 61, 60, 59, 58, 57,56,55$ \\
	& (5 weights) & (9 weights) \\
	\hline
	$[63, 7, 31]_2$ & $61, 60, 59, 58$ & $61, 60, 59, 58, 57,56,55$ \\
	& (4 weights) & (7 weights) \\
	\hline
	$[64, 6, 32]_2$ & $63, 62, 61, 60$ & $63, 62, 61, 60, 59, 58, 57$ \\
	& (4 weights) & (7 weights) \\
	\hline
	$[64, 7, 32]_2$ & $63, 62, 61, 60, 59$ & $63, 62, 61, 60, 59,58,57,56,55$ \\
	& (5 weights) & (9 weights) \\
	\hline
	$[65, 7, 32]_2$ & $63, 62, 61, 60$ & $63, 62, 61, 60, 59, 58, 57$ \\
	& (4 weights) & (7 weights) \\
	\hline
	$[71, 5, 36]_2$ & $71, 70, 69, 68$ & $71, 70, 69, 68, 67, 66, 65$ \\
	& (4 weights) & (7 weights) \\
	\hline
	$[75, 5, 38]_2$ & $75, 74, 73, 72$ & $75, 74, 73, 72, 71, 70, 69$ \\
	& (4 weights) & (7 weights) \\
	\hline
	$[77, 5, 39]_2$ & $77, 76, 75, 74$ & $77, 76, 75, 74, 73, 72, 71$ \\
	& (4 weights) & (7 weights) \\
	\hline
	$[78, 5, 40]_2$ & $78, 77, 76, 75$ & $79, 78, 77, 76, 75, 74, 73, 72, 71$ \\
	& (4 weights) & (9 weights) \\
	\hline
	$[79, 5, 40]_2$ & $79, 78, 77, 76$ & $79, 78, 77, 76, 75, 74, 73$ \\
	& (4 weights) & (7 weights) \\
	\hline
	$[80, 6, 40]_2$ & $79, 78, 77, 76$ & $79, 78, 77, 76, 75, 74, 73$ \\
	& (4 weights) & (7 weights) \\
	\hline
	$[83, 5, 42]_2$ & $83, 82, 81, 80$ & $83, 82, 81, 80, 79, 78, 77$ \\
	& (4 weights) & (7 weights) \\
	\hline
	$[85, 5, 43]_2$ & $85, 84, 83, 82$ & $85, 84, 83, 82, 81, 80, 79$ \\
	& (4 weights) & (7 weights) \\
	\hline
	$[86, 5, 44]_2$ & $86, 85, 84, 83$ & $87, 86, 85, 84, 83, 82, 81, 80, 79$ \\
	& (4 weights) & (9 weights) \\
	\hline
	$[87, 5, 44]_2$ & $87, 86, 85, 84$ & $87, 86, 85, 84, 83, 82, 81$ \\
	& (4 weights) & (7 weights) \\
	\hline
	$[88, 6, 44]_2$ & $87, 86, 85, 84$ & $87, 86, 85, 84, 83, 82, 81$ \\
	& (4 weights) & (7 weights) \\
	\hline
	$[89, 5, 45]_2$ & $89, 88, 87, 86$ & $89, 88, 87, 86, 85, 84, 83$ \\
	& (4 weights) & (7 weights) \\
	\hline
	$[90, 5, 46]_2$ & $90, 89, 88, 87$ & $90, 89, 88, 87, 86, 85, 84, 83$ \\
	& (4 weights) & (8 weights) \\
	\hline
	$[91, 5, 46]_2$ & $91, 90, 89, 88$ & $91, 90, 89, 88, 87, 86, 85$ \\
	& (4 weights) & (7 weights) \\
	\hline
	$[92, 5, 47]_2$ & $92, 91, 90, 89$ & $93, 92, 91, 90, 89, 88, 87, 86, 85$ \\
	& (4 weights) & (9 weights) \\
	\hline
	$[92, 6, 46]_2$ & $91, 90, 89, 88$ & $91, 90, 89, 88, 87, 86, 85$ \\
	& (4 weights) & (7 weights) \\
	\hline
	$[93, 5, 48]_2$ & $93, 92, 91, 90$ & $95, 94, 93, 92, 91, 90, 89, 88, 87,86,85$ \\
	& (4 weights) & (11 weights) \\
	\hline
\end{longtable}

\begin{table}[h]
	\centering
	\caption{Comparison of excluded weight ranges for ternary linear codes: Chen-Xie bound vs. Singleton Exclusion Criterion}
	\begin{tabular}{|p{85.9pt}|p{95.4pt}|p{180.4pt}|}
		\hline
	\textbf{Code parameters} & \textbf{Excluded weights} & \textbf{Excluded weights} \\
	$[n,k,d]_3$ & \textbf{(Chen-Xie bound)} & \textbf{(Singleton  Criterion)} \\
		\hline
		$[27, 4, 18]_3$ & $26, 25$ & $26, 25, 24, 23, 22$ \\
		& (2 weights) & (5 weights) \\
		\hline
		$[36, 4, 24]_3$ & $35, 34$ & $35, 34, 33, 32, 31$ \\
		& (2 weights) & (5 weights) \\
		\hline
		$[80, 4, 54]_3$ & $80, 79, 78$ & $80, 79, 78, 77, 76, 75, 74, 73$ \\
		& (3 weights) & (8 weights) \\
		\hline
		$[81, 5, 54]_3$ & $80, 79, 78$ & $80, 79, 78, 77, 76, 75, 74, 73$ \\
		& (3 weights) & (8 weights) \\
		\hline
		$[107, 4, 72]_3$ & $107, 106, 105$ & $107, 106, 105, 104, 103, 102, 101, 100$ \\
		& (3 weights) & (8 weights) \\
		\hline
		$[108, 5, 72]_3$ & $107, 106, 105$ & $107, 106, 105, 104, 103, 102, 101, 100$ \\
		& (3 weights) & (8 weights) \\
		\hline
		$[116, 4, 78]_3$ & $116, 115, 114$ & $116, 115, 114, 113, 112, 111, 110, 109$ \\
		& (3 weights) & (8 weights) \\
		\hline
		$[117, 5, 78]_3$ & $116, 115, 114$ & $116, 115, 114, 113, 112, 111, 110, 109$ \\
		& (3 weights) & (8 weights) \\
		\hline
		$[119, 4, 80]_3$ & $119, 118, 117$ & $119, 118, 117, 116, 115, 114, 113, 112$ \\
		& (3 weights) & (8 weights) \\
		\hline
		$[120, 4, 81]_3$ & $120, 119, 118$ & $121, 120, 119, 118, 117, 116, 115, 114, 113, 112$ \\
		& (3 weights) & (10 weights) \\
		\hline
		$[120, 5, 80]_3$ & $119, 118, 117$ & $119, 118, 117, 116, 115, 114, 113, 112$ \\
		& (3 weights) & (8 weights) \\
		\hline
		$[121, 5, 81]_3$ & $120, 119, 118$ & $121, 120, 119, 118, 117, 116, 115, 114, 113, 112$ \\
		& (3 weights) & (10 weights) \\
		\hline
		$[134, 4, 90]_3$ & $134, 133, 132$ & $134, 133, 132, 131, 130, 129, 128, 127$ \\
		& (3 weights) & (8 weights) \\
		\hline
		$[143, 4, 96]_3$ & $143, 142, 141$ & $143, 142, 141, 140, 139, 138, 137, 136$ \\
		& (3 weights) & (8 weights) \\
		\hline
		$[146, 4, 98]_3$ & $146, 145, 144$ & $146, 145, 144, 143, 142, 141, 140, 139$ \\
		& (3 weights) & (8 weights) \\
		\hline
		$[147, 4, 99]_3$ & $147, 146, 145$ & $148, 147, 146, 145, 144, 143, 142, 141, 140, 139$ \\
		& (3 weights) & (10 weights) \\
		\hline
		$[152, 4, 102]_3$ & $152, 151, 150$ & $152, 151, 150, 149, 148, 147, 146, 145$ \\
		& (3 weights) & (8 weights) \\
		\hline
		$[155, 4, 104]_3$ & $155, 154, 153$ & $155, 154, 153, 152, 151, 150, 149, 148$ \\
		& (3 weights) & (8 weights) \\
		\hline
		$[162, 5, 108]_3$ & $161, 160, 159$ & $161, 160, 159, 158, 157, 156, 155, 154$ \\
		& (3 weights) & (8 weights) \\
		\hline
		$[189, 5, 126]_3$ & $188, 187, 186$ & $188, 187, 186, 185, 184, 183, 182, 181$ \\
		& (3 weights) & (8 weights) \\
		\hline
		$[198, 5, 132]_3$ & $197, 196, 195$ & $197, 196, 195, 194, 193, 192, 191, 190$ \\
		& (3 weights) & (8 weights) \\
		\hline
		$[201, 4, 135]_3$ & $201, 200, 199$ & $202, 201, 200, 199, 198, 197, 196, 195, 194, 193$ \\
		& (3 weights) & (10 weights) \\
		\hline
		$[201, 5, 134]_3$ & $200, 199, 198$ & $200, 199, 198, 197, 196, 195, 194, 193$ \\
		& (3 weights) & (8 weights) \\
		\hline
		$[202, 5, 135]_3$ & $201, 200, 199$ & $202, 201, 200, 199, 198, 197, 196, 195, 194, 193$ \\
		& (3 weights) & (10 weights) \\
		\hline
	\end{tabular}
\end{table}

\subsection{Excluded weights via Theorem \ref{secondbound}}
By employing residual code techniques in conjunction with the Griesmer bound,
we derive refined constraints on excluded weights.

\begin{Theorem}{\rm (Griesmer Exclusion Criterion)}\label{firstvanishing3}
Let $\mathcal{C}$ be an $[n,k,d]$ linear code over $\mathbb{F}_q$ and
let $w$ be a positive integer. If $w$ satisfies the conditions $w < \frac{qd}{q-1}$ and
\[ n < d + \left\lceil \frac{w}{q} \right\rceil + \sum_{i=1}^{k-2} \left\lceil \frac{d - w + \left\lceil \frac{w}{q} \right\rceil}{q^i} \right\rceil, \]
then $\mathcal{C}$ contains no codeword of weight $w$.
\end{Theorem}
\begin{proof}
It follows directly from Theorem \ref{secondbound}.
\end{proof}

The following example illustrates that Theorem~\ref{firstvanishing3}
provides the most powerful criterion among Lemma~\ref{knownresult}
and Theorem~\ref{firstvanishing1} for determining excluded weights
of linear codes. A distinctive feature of Theorem~\ref{firstvanishing3},
which sets it apart from the other two results, is its ability to exclude
weights that do not necessarily lie within a consecutive interval.
To demonstrate this, we explicitly construct a binary linear code with parameters $[11,3,6]$,
whose set of nonzero weights is exactly $\{6,8\}$.
From the parameters $n,k,d,q$, Lemma~\ref{knownresult} guarantees the absence of codewords of weights $\{10,11\}$.
Applying Theorem~\ref{firstvanishing1} to the same parameters further excludes weights $\{9,10,11\}$.
Finally, Theorem~\ref{firstvanishing3} yields the vanishing set $\{7,9,10,11\}$.
This conclusion is optimal, since the remaining weights $\{6,8\}$ coincide precisely with the actual nonzero weights of the code.
Further details are provided in the following example.
\begin{Example}{\rm
Consider the binary linear code $\mathcal{C}$ with generator matrix
$$
G = \left[\begin{array}{ccccccccccc}
	1 & 1 & 1 & 1 & 1 & 1 & 1 & 1 & 0 & 0 & 0 \\
	1 & 1 & 1 & 1 & 0 & 0 & 0 & 0 & 1 & 1 & 0 \\
	1 & 1 & 0 & 0 & 1 & 1 & 0 & 0 & 1 & 0 & 1
\end{array}\right].
$$
It is straightforward to verify that $\mathcal{C}$ has parameters $[11,3,6]$ and consists of the following codewords:
$$
	\begin{aligned}
		\mathbf{0} = (0,0,0,0,0,0,0,0,0,0,0),
		&~~\mathbf{c}_1 = (1,1,1,1,1,1,1,1,0,0,0), \\
		\mathbf{c}_2 = (1,1,1,1,0,0,0,0,1,1,0),
		&~~\mathbf{c}_3 = (1,1,0,0,1,1,0,0,1,0,1), \\
		\mathbf{c}_4 = (0,0,0,0,1,1,1,1,1,1,0),
		&~~\mathbf{c}_5 = (0,0,1,1,0,0,1,1,1,0,1), \\
		\mathbf{c}_6 = (0,0,1,1,1,1,0,0,0,1,1),
		&~~\mathbf{c}_7 = (1,1,0,0,0,0,1,1,0,1,1).
	\end{aligned}
$$
	
It can be checked that $\mathcal{C}$ is a Griesmer code, i.e., it attains the Griesmer bound.
The only nonzero weights of $\mathcal{C}$ are $6$ and $8$.
We now analyze its excluded weights under three approaches:

- \textbf{Chen-Xie bound (Lemma \ref{knownresult})}:
  The condition $w > n - k + 1 = 9$ excludes weights $w = 10,11$.
	
- \textbf{Singleton Exclusion Criterion (Theorem \ref{firstvanishing1})}:
  The condition $w > q(n - k - d + 2)$ excludes weights $w = 9,10,11$.
	
- \textbf{Griesmer Exclusion Criterion  (Theorem \ref{firstvanishing3})}:
  The condition
  $$
  n < d + \left\lceil \frac{w}{q} \right\rceil + \sum_{i=1}^{k-2} \left\lceil \frac{d - w + \lceil \frac{w}{q} \rceil}{q^i} \right\rceil
  $$
  for $q=2$, $k=3$ specializes to
  $$
  11 < 6 + \left\lceil \frac{w}{2} \right\rceil + \left\lceil \frac{6 - w + \lceil \frac{w}{2} \rceil}{2} \right\rceil,
  $$
  which excludes weights $w = 7,9,10,11$.

This demonstrates a progressive expansion in the set of excluded weights:
$$
\text{Chen-Xie: } \{10,11\} \;\subsetneq\; \text{Singleton: } \{9,10,11\} \;\subsetneq\; \text{Griesmer: } \{7,9,10,11\}.
$$
Note that the actual code $\mathcal{C}$ indeed has no codewords of weight $7,9,10,11$.
Thus, the Griesmer Exclusion Criterion correctly identifies {\bf all and only} the  excluded weights, in perfect agreement with the true weight distribution.
}
\end{Example}

To further highlight the effectiveness of our Griesmer Exclusion Criterion
compared with both the Chen-Xie bound and our Singleton Exclusion Criterion, we evaluate the number of excluded weights identified by each method for several Griesmer codes listed in Table I of Chen and Xie's paper \cite{CX}. The comparison is summarized in the following table:
\begin{table}[h]
	\centering
	\small
	\caption{Comparison of three exclusion methods for Griesmer codes:
		Chen-Xie bound, Singleton Exclusion Criterion and Griesmer Exclusion Criterion}
	\begin{tabular}{|p{80.9pt}|p{1.8cm}|p{1.8cm}|p{7cm}|}
		\hline
			\textbf{Code parameters} & \textbf{Excluded weights} & \textbf{Excluded weights} & \textbf{Excluded weights}\\
		$[n,k,d]_3$ & \textbf{(Chen-Xie bound)} & \textbf{(Singleton  Criterion)} & \textbf{(Griesmer  Criterion)}\\
		\hline
		$[267,8,132]_2$ & 261-263 & 259-263 & 133-135, 167, 183, 191, 195, 197-199, 215, 223, 227, 229-231, 239, 243, 245-247, 251, 253-255, 257-263 \\
		& (3 weights) & (5 weights) & (32 weights) \\
		\hline
		$[271,8,134]_2$ & 265-267 & 263-267 & 135, 137-139, 171, 187, 195, 199, 201-203, 219, 227, 231, 233-235, 243, 247, 249-251, 255, 257-259, 261-267 \\
		& (3 weights) & (5 weights) & (33 weights) \\
		\hline
		$[274,8,136]_2$ & 268-271 & 265-271 & 137-143, 159, 167, 171, 173-175, 183, 187, 189-191, 195, 197-199, 201-207, 215, 219, 221-223, 227, 229-231, 233-239, 243, 245-247, 249-255, 257-271 \\
		& (4 weights) & (7 weights) & (71 weights) \\
		\hline
		$[279,8,138]_2$ & 273-275 & 271-275 & 139, 143, 145-147, 179, 195, 203, 207, 209-211, 227, 235, 239, 241-243, 251, 255, 257-259, 263, 265-267, 269-275 \\
		& (3 weights) & (5 weights) & (34 weights) \\
		\hline
		$[282,8,140]_2$ & 276-279 & 273-279 & 141-143, 145-151, 167, 175, 179, 181-183, 191, 195, 197-199, 203, 205-207, 209-215, 223, 227, 229-231, 235, 237-239, 241-247, 251, 253-255, 257-263, 265-279 \\
		& (4 weights) & (7 weights) & (79 weights) \\
		\hline
		$[286,8,142]_2$ & 280-283 & 277-283 & 143, 145-147, 149-155, 171, 179, 183, 185-187, 195, 199, 201-203, 207, 209-211, 213-219, 227, 231, 233-235, 239, 241-243, 245-251, 255, 257-259, 261-267, 269-283 \\
		& (4 weights) & (7 weights) & (83 weights) \\
		\hline
		$[289,8,144]_2$ & 283-287 & 279-287 & 145-159, 167, 171, 173-175, 179, 181-183, 185-191, 195, 197-199, 201-207, 209-215, 216-223, 227, 229-231, 233-239, 241-247, 248-255, 257-263, 264-271, 272-279, 280-287 \\
		& (5 weights) & (9 weights) & (143 weights) \\
		\hline
	\end{tabular}
\end{table}

\section{Conclusion and future works}
We introduced a family of \emph{$w$-aware bounds} for $q$-ary linear codes that leverage the existence of a prescribed nonzero codeword of weight $w$. By passing to residual codes with respect to that codeword and applying classical inequalities, we obtained:

\begin{itemize}
  \item  \textbf{Residual-Singleton bound}
  \[
  d \le n-k-\Big\lceil\tfrac{w}{q}\Big\rceil+2 \qquad \text{for } 0<w<\tfrac{qd}{q-1},
  \]
  \item  \textbf{Residual-Griesmer bound}
  \[
  n \ge d+\Big\lceil\tfrac{w}{q}\Big\rceil+\sum_{i=1}^{k-2}\Big\lceil \tfrac{\,d-w+\lceil \frac{w}{q}\rceil\,}{q^i}\Big\rceil,
  \qquad \text{for } 0<w<\tfrac{qd}{q-1},
  \]
  \item and an elementary \textbf{global weight bound} $w\le q(n-d)$, which yields $d\le \tfrac{q}{q+1}n$ upon taking $w=d$.
\end{itemize}

These inequalities give \emph{explicit vanishing windows} for the weight spectrum and sharpen structural constraints that are invisible to bounds depending only on $(n,k,d,q)$. In particular, for MDS codes we proved that any codeword with $d\le w<\tfrac{qd}{q-1}$ must satisfy $w\le q$, revealing a pronounced gap in their weight distributions. The bounds are tight on several standard families (e.g., Griesmer codes), and a detailed $[11,3,6]_2$ example shows that the Griesmer Exclusion Criterion can rule out \emph{non-consecutive} weights and exactly recover the true spectrum where the Chen-Xie bound only yields shorter consecutive intervals. Numerical comparisons across standard parameter sets confirm that the proposed bounds \emph{strictly enlarge} known excluded weight ranges.

While our approach is elementary---combining residual codes with classical bounds---it exposes a versatile mechanism for translating a single structural datum (one known weight $w$) into concrete, and often stronger, parameter trade-offs. We expect this perspective to be useful both for code design and for analyzing the combinatorial structure of linear codes and their duals.


\begin{thebibliography}{21}


\bibitem{AB}
A. Ashikhmin and A. Barg, ``Minimal vectors in linear codes'',
{\em IEEE Transactions on Information Theory}, vol. 44, no. 5, pp. 2010-2017, 1998.



\bibitem{Ca}
V. R. Cadambe and A. Mazumdar, ``Bounds on the size of locally recoverable codes'',
{\em IEEE Transactions on Information Theory}, vol. 61, no. 8, pp. 5787-5794, 2015.

\bibitem{CX}
H. Chen and C. Xie, ``A new upper bound for linear codes and vanishing partial weight distributions'',
{\em IEEE Transactions on Information Theory}, vol. 70, no. 12, pp. 8713-8722, 2024.

\bibitem{CZ}
B. Chen and G. Zhang, ``A tight upper bound on the number of non-zero weights of a cyclic code'',
{\em IEEE Transactions on Information Theory}, vol. 69, no. 2, pp. 995-1004, 2023.


\bibitem{Delsarte}
P. Delsarte, ``Four fundamental parameters of a code and their combinatorial significance'', {\em Informatioin and Control},
vol. 23, no. 5, pp. 407-438, 1973.

\bibitem{Ding2018}
C. Ding and  C. Tang, Designs from Linear Codes, World Scientific, Singapore, 2018.



\bibitem{Go}
P. Gopalan, C. Huang, H. Simitci, and S. Yekhanin, ``On the locality
of codeword symbols'', {\em IEEE Transactions on Information Theory,}
vol. 58, no. 11, pp. 6925-6934,   2012.


\bibitem{Griesmer}
J. H. Griesmer, ``A bound for error-correcting codes'', {\em IBM J. Research Develop.}, vol. 4, no. 5, pp. 532-542, 1960.

\bibitem{HS}
H. J. Helgert and R. D. Stinaff, ``Minimum distance bounds for binary linear codes'',
{\em IEEE Transactions on Information Theory}, vol. 19, no. 3, pp. 344-356, 1973.


\bibitem{HN}
R. Hill and D. E. Newton, ``Optimal ternary linear codes'',
{\em Designs, Codes and Cryptography}, vol.2,  no.2,  pp. 137-157, 1992.

\bibitem{huffman2003}
W. C. Huffman, V. Pless,  Fundamentals of Error-Correcting Codes. Cambridge: Cambridge University Press, 2003.



\bibitem{LCEL}
G. Luo, X. Cao, M. F. Ezerman and S. Ling, ``On the weights of linear codes with prescribed automorphisms'',
{\em IEEE Transactions on Information Theory}, vol. 70, no. 7, pp. 4983-4989, 2024.

\bibitem{Luo25}
G. Luo, B. Chen, M. F. Ezerman and S. Ling,  ``Bounds and constructions of quantum locally recoverable codes from quantum CSS codes'',
{\em IEEE Transactions on Information Theory}, vol. 71, no. 3, pp. 1794-1802, 2025.


\bibitem{Pra}
N. Prakash, G. M. Kamath, V. Lalitha, and P. V. Kumar,
``Optimal linear codes with a local-error-correction property'',   in Proc. Int. Symp.
Inf. Theory (ISIT), Cambridge, MA, U.S.A.,  pp. 2776-2780, Jul. 2012.

\bibitem{Ra}
A. S. Rawat, D. S. Papailiopoulos  and A. G. Dimakis, ``Locality and availability in distributed storage'',
{\em IEEE Transactions on Information Theory}, vol. 62, no. 8, pp. 4481-4493,
 2016.

\bibitem{Shi}
M. Shi, A. Neri, and P. Sol\'{e}, ``How many weights can a quasi-cyclic
code have?'' {\em IEEE Transactions on Information Theory}, vol. 66, no. 11, pp. 6855-6862,
 2020.

\bibitem{Singleton}
R. C. Singleton, ``Maximum distance $q$-ary codes'',
{\em IEEE Transactions on Information Theory}, vol. 10, no. 2, pp. 116-118, 1964.


\end{thebibliography}
\end{document}